\documentclass[review]{elsarticle}
\usepackage[latin1]{inputenc}
\usepackage{url,amsfonts,epsfig}
\usepackage{amsmath,latexsym,amssymb,txfonts}
\usepackage{color,caption}
\usepackage{graphicx, subfigure}
\usepackage{footmisc,url}
\usepackage{algorithm,algorithmic}

\newtheorem{theorem}{Theorem}

\newtheorem{proposition}{Proposition}
\newtheorem{lemma}{Lemma}

\newproof{proof}{Proof}

\newcommand{\ls}[1]
    {\dimen0=\fontdimen6\the\font
     \lineskip=#1\dimen0
     \advance\lineskip.5\fontdimen5\the\font
     \advance\lineskip-\dimen0
     \lineskiplimit=.9\lineskip
     \baselineskip=\lineskip
     \advance\baselineskip\dimen0
     \normallineskip\lineskip
     \normallineskiplimit\lineskiplimit
     \normalbaselineskip\baselineskip
     \ignorespaces}
\let\oldmarginpar\marginpar
\renewcommand\marginpar[1]{ {\ls{1}
\-\oldmarginpar[\raggedleft\footnotesize\textcolor{red}{#1}]%
{\raggedright\footnotesize\textcolor{red}{#1}}
}}

\begin{document}

\begin{frontmatter}

\title{All graphs with at most seven vertices are Pairwise Compatibility Graphs}

\author[rvt]{T. Calamoneri\corref{cor1}\fnref{fn1}}
\ead{calamo@di.uniroma1.it}

\author[rvt]{D. Frascaria\corref{cor1}\fnref{fn1}}
\ead{dariio@msn.com}

\author[rvt]{B. Sinaimeri\corref{cor1}\fnref{fn1}}
\ead{sinaimeri@di.uniroma1.it}

\address[rvt]{Department of Computer Science, ``Sapienza'' University of Rome, Via Salaria 113, 00198 Roma, Italy}


\begin{abstract}
A graph $G$ is called a pairwise compatibility graph (PCG) if there exists an edge-weighted tree $T$ and two non-negative real numbers $d_{min}$ and $d_{max}$ such that each leaf $l_u$ of $T$ corresponds to a vertex $u \in V$ and there is an edge $(u,v) \in E$ if and only if $d_{min} \leq d_{T,w} (l_u, l_v) \leq d_{max}$ where $d_{T,w} (l_u, l_v)$ is the sum of the weights of the edges on the unique path from $l_u$ to $l_v$ in $T$.  

In this note, we show that  all the graphs with at most seven vertices are PCGs.  In particular all these graphs  exept for the wheel on $7$ vertices $W_7$ are PCGs of a particular structure of a tree: a centipede. 
\end{abstract}


\end{frontmatter}
\section{Introduction}

A graph $G=(V,E)$ is a {\em pairwise compatibility graph} (PCG) if there exists a tree $T$, an edge-weight function  $w$ that assigns positive values to the edges of $T$ and  two non-negative real numbers $d_{min}$ and $d_{max}$, with $d_{min} \leq d_{max}$, such that each vertex $u \in V$ is uniquely associated to a leaf $l_u$ of $T$ and there is an edge $(u,v) \in E$ if and only if $d_{min} \leq d_{T,w} (l_u, l_v) \leq d_{max}$ where $d_{T,w} (l_u, l_v)$ is the sum of the weights of the edges on the unique path from $l_u$ to $l_v$ in $T$. In such a case, we say that $G$ is a PCG of $T$ for $d_{min}$ and $d_{max}$; in symbols, $G=PCG(T, w, d_{min}, d_{max})$.

It is clear that if a tree $T$, an edge-weight function $w$ and two values $d_{min}$ and $d_{max}$ are given, the construction of a PCG is a trivial problem. 
We focus on the reverse of this problem, i.e., given a graph $G$ we have to find out a tree $T$, an edge-weight function $w$ and suitable values, $d_{min}$ and $d_{max}$, such that $G=PCG(T,w,d_{min}, d_{max})$. Such a problem is called the {\em pairwise compatibility tree construction problem}.  

The concept of pairwise compatibility was introduced in \cite{Kal03} in a computational biology context and the weight function $w$ has positive values, as it represents a not null distance.  There are several known specific graph classes of pairwise compatibility graphs, e.g., cliques and disjoint union of cliques \cite{B}, chordless cycles and single chord cycles \cite{YHR09}, some  particular subclasses of bipartite graphs \cite{YBR10}, some particular subclasses of split matrogenic graphs \cite{CPS12}.  Furthermore a lot of work has been done concerning some particular subclasses of PCG as leaf power graphs \cite{B}, exact leaf power graphs \cite{BLR10} and lately a new subclass has been introduced, namly the min-leaf power graphs \cite{CPS12}.

Initially, the authors of \cite{Kal03} conjectured that every graph is a PCG, but this conjecture has been confuted in \cite{YBR10}, where a particular bipartite graph with 15 nodes has been proved not to be a PCG. 
This latter result has given rise to this research as it is natural to ask for the smallest graph that is not a PCG.

\medskip

A {\em caterpillar} $\Gamma_n$ is an $n$-leaf tree for which any leaf is at a distance exactly one from a central path called {\em spine}.
A {\em centipede} is an $n$-leaf caterpillar, in which the edges incident to the leaves produce a perfect matching.
Deleting from an $n$-leaf centipede the degree two vertices and merging the two edges incident to each of these vertices into a unique edge, results in a new caterpillar that we will call {\em reduced centipede} and denote by $\Pi_n$  (as an example, $\Pi_5$ is depicted at the top left of  Fig. \ref{fig.5nodes}).

Caterpillars are interesting trees in the context of PCGs, as in most of the cases, the pairwise compatibility tree construction problem admits as solution a tree that is in fact a caterpillar.
For this reason, we focus on this special kind of tree.
In this note, we prove that all the graphs with at most seven vertices are PCGs. More precisely, we demonstrate the following results:

\begin{itemize}
\item
If $G=PCG(\Gamma_n,w,d_{min}, d_{max})$, then there always exist a new edge-weight function $w'$, and a new value $d'_{max}$ such that it also holds: $G=PCG(\Pi_n,w',d_{min}, d'_{max})$. 

\item
It is well known that graphs with five vertices or less are all PCGs and the witness trees -- not all caterpillars -- are shown in \cite{P02}.
For each one of these graphs we prove that it is PCG of a reduced centipede, providing accordingly, an edge-weight function $w$ and the two values $d_{min}$ and $d_{max}$.

\item
All the graphs with six and seven vertices, except for the wheel $W_7$ (i.e. the graph formed by connecting a single vertex to all vertices of a cycle of length six -- see Figure \ref{fig.wheel}.a), are PCGs of a reduced centipede and, for each of them, we provide the edge-weight function $w$ and the two values $d_{min}$ and $d_{max}$ such that it is $PCG(\Pi_n, w, d_{min}, d_{max})$, $n=6, 7$.

\item
For what concerns the wheel $W_7$, it is known \cite{CFS} that $W_7$ is not PCG of the reduced centipede $\Pi_7$ (and hence it is not PCG of a caterpillar). 
We show that $W_7$ is PCG of a tree different from a caterpillar.
\end{itemize}

\section{Preliminaries}

In this section we list some results that will turn out to be useful in the rest of the paper.

Let $T$ be a tree such that there exist an edge-weight function $w$ and two non-negative values $d_{min}$ and $d_{max}$ such that $G=PCG(T,w,d_{min}, d_{max})$.
Observe that if $T$ has at least $4$ vertices and contains a vertex $v$ of degree $2$, then we can construct a new tree $T'$ in which $v$ is eliminated, the two edges $(x,v)$ and $(v,y)$ incident to $v$ are merged into a unique edge $(x,y)$ and a new function $w'$ is defined from $w$ only modifying the weight of the new edge, that is set equal to the sum of the weights of the old edges: $w'((x,y))=w((x,v))+w((v,y))$.
It is easy to see that $G=PCG(T',w',d_{min}, d_{max})$. For this reason, from now on, we will assume that all the trees we handle do not contain vertices of degree two. 

\begin{proposition} \cite{CMPS}
\label{prop.integer}
Let $G=PCG(T,w, d_{min},d_{max})$, where $d_{min}, d_{max}$ and 
the weight $w(e)$ of each edge $e$ of $T$ are positive
real numbers.
Then it is possible to choose $\hat{w},\hat{d}_{min},\hat{d}_{max}$ such that for any $e$, the quantities $\hat{d}_{min},\hat{d}_{max}$ and $\hat{w}(e)$  are natural numbers and $G=PCG(T,\hat{w},\hat{d}_{min},\hat{d}_{max})$.
\end{proposition}

We prove here the following useful lemma:
\begin{lemma}
\label{lemma.1}
Let $G=PCG(T,w, d_{min},d_{max})$. It is possible to choose $\hat{w},\hat{d}_{min},\hat{d}_{max}$ such that $\min \hat{w}(e)=1$, where the minimum is computed on all the edges of $T$, and  $G=PCG(T,\hat{w},\hat{d}_{min},\hat{d}_{max})$.
\end{lemma}

\begin{proof}
According to Proposition \ref{prop.integer}, we can assume that the edge weight $w$ and the two values $d_{min},d_{max}$ are integers.
Let $e_1, \ldots , e_n$ be the edges of $T$ incident to the leaves.
Without loss of generality, we can assume $w(e_1)=\min_i w(e_i)$.

We define $\hat{w}$ as follows:
$\hat{w}(e_1)=1$ and for each $i=2, \ldots , n$ define $\hat{w}(e_i)=w(e_i)-w(e_1)+1$.
Clearly, the function $\hat{w}$ is well defined as all its values are positive.

As the weight of any edge incident to a leaf has been decreased by exactly $w_1-1$ and the rest of the weights remained unchanged, then for of any two leaves $l_i,l_j$ it holds that  $d_{T, \hat{w}}(l_i,l_j)=d_{T, w}(l_i,l_j)-2w(e_1)+2$. 
Let $\hat{d}_{min}=\mbox{max}\{d_{min}-2w(e_1)+2, 0 \}$ and $\hat{d}_{max}=d_{max}-2w(e_1)+2$. 
It is easy to see that $G=PCG(T,\hat{w},\hat{d}_{min},\hat{d}_{max})$ indeed, if $\hat{d}_{min}=0$ then it means that there was no path weight below $d_{min}$, with respect to $w$. \qed
\end{proof}

The previous results imply that it is not restrictive to assume that the weights and $d_{min}$ and $d_{max}$ are integers and that the smallest weight is $1$.  Thus, in the rest of the paper we  will use these assumptions.

\section{PCGs of Caterpillars}

In this section we will prove that we can get rid of different kinds of caterpillar structures and restrict to consider only reduced centipedes. 

\begin{theorem}
\label{th.caterpillar}
Let $G$ be an $n$ vertex graph, $\Gamma_n$ and $\Pi_n$ be an $n$-leaf caterpillar without degree 2 vertices and an $n$-leaf reduced centipede, respectively.

Let $G=PCG(\Gamma_n,w,d_{min}, d_{max})$.
It is possible to choose $w'$ and  $d'_{max}$ such that $G=PCG(\Pi_n,w',d_{min}, d'_{max})$. 
\end{theorem}

\begin{proof}
In order not to overburden the exposition, let $\Gamma = \Gamma_n$ and $\Pi = \Pi_n$.

If $\Gamma$ is a reduced centipede, the claim is trivially proved, so assume it is not.
We lead the proof into two steps. 
First we define a non-negative edge-weight function $w''$ proving that $\Gamma$ weighted by $w$ and $\Pi$ weighted by $w''$ generate the same PCG $G$ for the same values $d_{min}$ and  $d_{max}$.
Then we modify $w''$ into a positive weight function $w'$ and introduce two new values $d'_{min}$ and $d'_{max}$ proving that $G$ is also $PCG(\Pi,w',d'_{min}, d'_{max})$.

Draw $\Gamma$ so that: i) the spine lies on a horizontal line, ii) all the leaves lie on a parallel line and iii) the edges between the spine and the leaves are represented as non-crossing line segments; number the leaves and the vertices of the spine from left to right $l_1, \ldots , l_n$ and $s_1, \ldots s_k$, $k <n$, respectively.
By drawing the reduced centipede $\Pi$ in a similar way, we number the leaves and the vertices of the spine from left to right by $m_1, \ldots , m_n$ and $t_2, \ldots t_{n-1}$.

We define the edge-weight function $w''$ as follows:

\begin{itemize}
\item
let $p(l_i)$ the unique adjacent vertex of $l_i$ in $\Gamma$; for each $1 < i < n$, define $w''((m_i, t_i))=w((l_i, p(l_i)))$;
\item
define $w''((m_1, t_2))=w((l_1, p(l_1)))$ and $w''((m_n, t_{n-1}))=w((l_n, p(l_n)))$;
\item
for each $2\leq i \leq n-2$, define $w''((t_i, t_{i+1}))=0$ if and only if $p(l_i)=p(l_{i+1})$ in $\Gamma$;
\item
for each $2\leq i \leq n-2$, define $w''((t_i, t_{i+1}))=w((p(t_i), p(t_{i+1})))$ if and only if $p(l_i) \neq p(l_{i+1})$ in $\Gamma$.
\end{itemize}

Observe that $w''$ is well defined, as $\Gamma$ has no degree 2 vertices.
 
It is quite easy to convince oneself that for each pair of leaves in $\Gamma$, $l_i$ and $l_j$, $d_{\Gamma, w}(l_i,l_j)$ is exactly the same as $d_{\Pi, w''}(m_i, m_j)$ and that $d_{min}$ and  $d_{max}$ remain unchanged, so $G=PCG(\Gamma, w'', d_{min}, d_{max})$. 

\medskip

It remains to show that we can reassign the edge-weights of $\Pi$ in a way that any edge gets a positive weight and  $\Pi$ is the pairwise compatibility tree of $G$. To this purpose, we denote by $E(H)$ the edge set of any graph $H$, and we introduce the following two quantities:
$$
L=\min_{(u,v)\not\in
E(G)}\left\{| d_{min}-d_{\Pi,w''}(l_{u},l_{v}) |,| d_{max}-d_{\Pi,w''}(l_{u},l_{v}) |\right\},
\qquad N=|\left\{e:  e \in E(\Pi), w(e)=0 \right\}|,
$$

$L$ is the smallest distance between the quantities $d_{min},d_{max}$ and the weighted distances on the tree of the
paths  corresponding to non-edges of $G$; 
$N$ is the number of edges of $\Pi$ of weight $0$. 

Observe that, unless $G$ coincides with the clique $K_n$ (which trivially is PCG of the reduced centipede), there always exists a pair of leaves such that their distance on $\Pi$ falls out of the interval $[d_{min}, d_{max}]$ and hence $L>0$. 
Furthermore, as any edge incident to a leaf in $\Pi$ is strictly greater than $0$ and in view of the hypothesis that the caterpillar $\Gamma$ is not a reduced caterpillar, it holds 
$1 \leq N \leq n-3$ (the bound $n-3$ is reached when $\Gamma$ is a star).
So, the value $\epsilon=\frac{L}{N+1}$ is well defined.

Now define a new weight function $w'$ on $\Pi$ by assigning the weight $\epsilon$ to any edge of weight $0$. More formally, $w'(e)=w''(e)$ if $w''(e) \neq 0$ and $w'(e)=\epsilon$ otherwise. 
In this way the distance between any two leaves in $\Pi$ can result increased by a value upper bounded by $\epsilon N < L$. 

Set the new value $d'_{max}=d_{max}+\epsilon N$.

The following three observations conclude the proof:
\begin{itemize}
\item
any distance between leaves in $\Pi$ that was strictly smaller than $d_{min}$ with respect to the weight function $w''$ remains so after this transformation in view of the fact that $\epsilon N < L$;
\item
any distance that was strictly greater than $d_{max}$ with respect to the weight function $w''$ is strictly greater than $d'_{max}$ due to the definition of $L$;
\item
any distance that was in the interval  $[d_{min}, d_{max}]$ with respect to the weight function $w''$ is now in the interval  $[d_{min}, d'_{max}]$. \qed
\end{itemize}
\end{proof}

Observe that the previous statement suggests not to consider all kinds of caterpillars, but to restrict to reduced centipedes, only.
In the next section we exploit this result.

\begin{figure}[!ht]
\centering
\includegraphics[width = \textwidth]{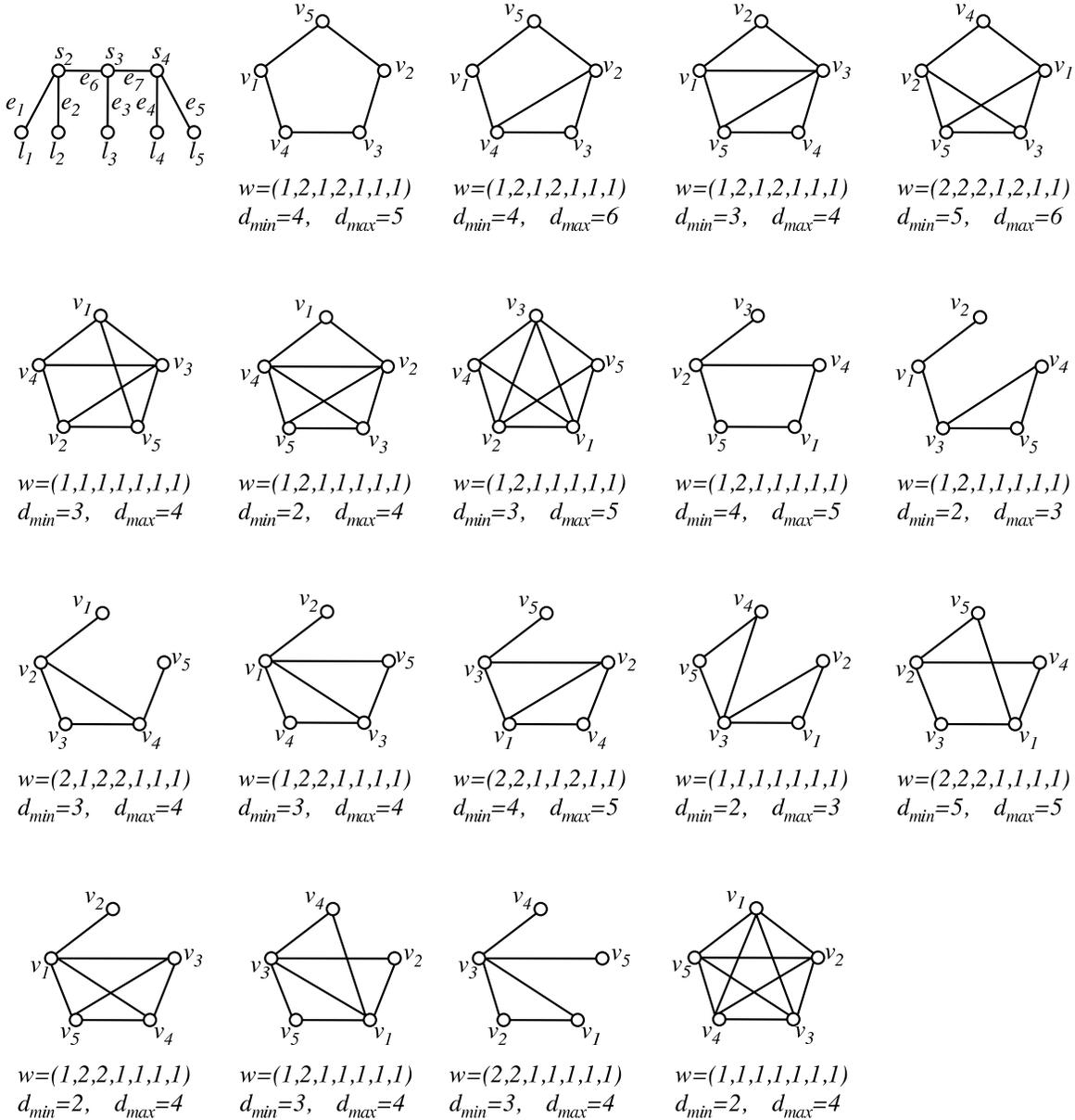}
\caption{All the non isomorphic connected cyclic graphs with 5 vertices with their representation as PCGs of the reduced centipede (top left).} \label{fig.5nodes}
\end{figure}

\section{Graphs on at most seven vertices}

In this section we show that all graphs with at most seven vertices, except for the wheel $W_7$, are PCGs of a  reduced centipede. 

Analogously to what we did in the proof of Theorem \ref{th.caterpillar}, name the leaves of $\Pi_n$ from left to right with $l_1, \ldots , l_n$ and the vertices of the spine from left to right with $s_2, \ldots s_{n-1}$.
As, for any $n$, there exists a unique unlabeled reduced centipede with $n$ leaves $\Pi_n$, in the following we consider the edges of $\Pi_n$ as ordered in the following way:
$e_1= (l_1, s_2)$; $e_i=(l_i, s_i)$ for each $2 \leq i \leq n-1$; $e_n=(l_n, s_{n-1})$; finally, $e_{n+i-1}=(s_i, s_{i+1})$ for each $2 \leq i \leq n-2$.

Now, the edge-weight function $w$ can be expressed as a $(2n-3)$ long vector $\vec{w}$, where the component $w_i$ is a positive integer representing the weight assigned to edge $e_i$.

In Figure \ref{fig.5nodes} all the 18 connected non isomorphic cyclic graphs with 5 vertices are depicted, together with the vector $\vec{w}$ and the values of $d_{min}$ and $d_{max}$ that witness that all of them are PCGs of $\Pi_5$.  Observe that the connected non isomorphic graphs on 5 vertices are 21, we have omitted the 3 graphs that are trees, which are trivially PCGs.
We remind that it is already proved in \cite{P02} that all the graphs with at most five vertices are PCG, but the provided trees were all different and not all caterpillars.

\medskip

For what concerns graphs with 6 and 7 vertices, except for the wheel $W_7$, we get a similar result. 
For the sake of brevity we do not depict all these graphs (there are 112 connected non isomorphic graphs with 6 vertices and 853 with 7 vertices), but the values of $\vec{w}$, $d_{min}$ and $d_{max}$ we got with the help of an enumerative C program can be found at the web page \url{https://sites.google.com/site/pcg6and7vertices/} .\\
Thus, we obtain the following result:

\begin{lemma}
\label{lemma:567}
All graphs with at most 7 vertices except for the wheel $W_7$ are PCGs of a reduced centipede.
\end{lemma}

\begin{lemma}
\label{lemma:wheel}
The graph $W_7$ is a PCG.
\end{lemma}
\begin{proof}
Consider the edge-weighted tree $T$ depicted in Figure \ref{fig.wheel}.b and the two values $d_{min}=5$ and $d_{max}=7$.
It is immediate to see that $W_7=PCG(T, w, d_{min}, d_{max})$. \qed
\end{proof}

\begin{figure}[!ht]
\centering
\includegraphics[width = 0.7\textwidth]{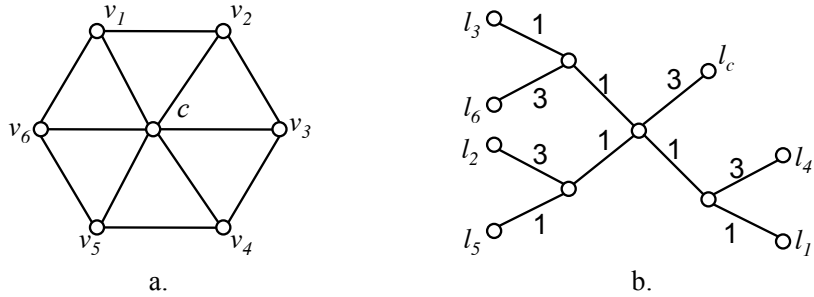}
\caption{(a.) The wheel $W_7$ and (b.) the edge-weighted tree $T$ such that $W_7=PCG(T, w, 5, 7)$.} 
\label{fig.wheel}
\end{figure}

This result is in agreement with the negative result in \cite{CFS}, stating that it is not possible to find any edge-weight function $w$ and any two values $d_{min}$ and $d_{max}$ such that $W_7=PCG(\Pi_7, w, d_{min}, d_{max})$.

\medskip

From Lemmas \ref{lemma:567} and \ref{lemma:wheel} it immediately derives the main result of this note:

\begin{theorem}
All graphs with at most 7 vertices are PCGs.
\end{theorem}


\end{document}